\documentclass[11pt]{article}
\usepackage[colorlinks]{hyperref}
\hypersetup{linkcolor=cyan,filecolor=cyan,citecolor=cyan,urlcolor=cyan}
\usepackage{xspace}
\usepackage{thm-restate,color,xcolor}
\usepackage{amsmath,amssymb,bbm,amsthm}
\usepackage{fullpage}
\usepackage{boxedminipage}
\usepackage[boxed]{algorithm}
\usepackage{epigraph}
\usepackage[sc]{mathpazo}
\usepackage{amsmath}
\usepackage{mathtools}
\usepackage{framed}
\usepackage[framemethod=tikz]{mdframed}
\usepackage{titlesec}
\usepackage{lipsum}%
\usepackage{cleveref,aliascnt}
\usepackage{tikz}
\usepackage{wrapfig}

\usetikzlibrary{shapes.geometric}
\usetikzlibrary{arrows}
\usetikzlibrary{arrows.meta}
\usetikzlibrary{patterns}
\usetikzlibrary{shapes.misc}

\newcommand{\congest}{${\mathsf{CONGEST}}$}

\newcommand{\dist}{\mbox{\rm dist}}

\newtheorem{theorem}{Theorem}[section]
\newtheorem{lemma}[theorem]{Lemma}
\newtheorem{meta-theorem}[theorem]{Meta-Theorem}
\newtheorem{claim}[theorem]{Claim}

\newtheorem{corollary}[theorem]{Corollary}

\newtheorem{observation}[theorem]{Observation}
\newtheorem{definition}[theorem]{Definition}
\newtheorem{fact}[theorem]{Fact}

\newcommand{\dilation}{\mbox{\tt d}}
\newcommand{\congestion}{\mbox{\tt c}}

\renewcommand{\paragraph}[1]{\vspace{0.15cm}\noindent {\bf #1}}

\def\FTBFS{\mbox{\tt FT-BFS}}

\def\FTMBFS{\mbox{\tt FT-MBFS}}
\def\LastE{\mbox{\tt LastE}}
\def\Sample{\mbox{\tt Sample}}
\def\Rel{\pi_{\sigma'}}
\def\SubRel{\pi_{\sigma}}
\def\BFS{\mbox{\tt BFS}}

\title{Distributed Constructions of Dual-Failure Fault-Tolerant Distance Preservers}
\date{}
\author{
Merav Parter\\
	\small Weizmann Institute \\
	\small merav.parter@weizmann.ac.il \thanks{Partially supported by the Israeli Science Foundation (ISF), grant no. 2084/18.}
}

\begin{document}
\date{}

\maketitle
\begin{abstract}
Fault tolerant distance preservers (spanners) are sparse subgraphs that preserve (approximate) distances between given pairs of vertices under edge or vertex failures.  So-far, these structures have been studied mainly from a
centralized viewpoint. Despite the fact fault tolerant preservers are mainly motivated by the error-prone nature of distributed networks, not much is known on the distributed computational aspects of these structures. 

In this paper, we present distributed algorithms for constructing fault tolerant distance preservers and $+2$ additive spanners that are resilient to at most \emph{two edge} faults. Prior to our work, the only non-trivial constructions known were for the \emph{single} fault and \emph{single source} setting by [Ghaffari and Parter SPAA'16]. 

Our key technical contribution is a distributed algorithm for computing distance preservers w.r.t. a subset $S$ of source vertices, resilient to two edge faults. The output structure contains a BFS tree $BFS(s,G \setminus \{e_1,e_2\})$ for every $s \in S$ and every $e_1,e_2 \in G$. The distributed construction of this structure is based on a delicate balance between the edge congestion (formed by running multiple BFS trees simultaneously) and the sparsity of the output subgraph. No sublinear-round algorithms for constructing these structures have been known before. 

\end{abstract}

\section{Introduction}
Fault tolerant distance preservers are sparse subgraphs that preserve distances between given pairs of nodes
under edge or vertex failures. In this paper, we present the first non-trivial distributed constructions of source-wise distance preservers and additive spanners that can handle \emph{two} edge failures. 
We start by providing some background on fault-tolerant preservers from a graph-theoretical perspective, and then provide the distributed algorithmic context. 

\paragraph{Fault-Tolerant Distance Preserves.} Distances preservers are sparse subgraphs that preserve the distances between a given pairs of nodes in an \emph{exact} manner. As distance preservers are often computed for distributed networks where parts can spontaneously fail, fault-tolerance is a desired requirement for these structures. For every bounded set of edge failures, the fault tolerant preservers are required to contain
\emph{replacement paths} around the faulted set. Formally, for a pair of vertices $s$ and $t$ and a subset of edge failures $F$, a replacement path $P(s,t,F)$ is an $s$-$t$ shortest path in the surviving graph $G \setminus F$. 
The efficient (centralized) computation of all replacement path distances for a given $s$-$t$ pair and a given source vertex $s$ has attracted a lot of attention since the 80's \cite{nardelli1997low,nardelli2003finding,emek2010near,roditty2012replacement,grandoni2012improved,weimann2013replacement,ChechikC19,AlonCC19,Check20}.
Most of these works focus on the single-failure case, and relatively little is known on the complexity of distance preserving computation under multiple edge faults.

Parter and Peleg \cite{ParterP16} introduced the notion of \FTBFS\ structures with respect to given source vertex $s$. Roughly speaking, an \FTBFS\ structure is a subgraph of the original graph that preserves all $\{s\}\times V$ distances under a single failure of an edge or a vertex. An \FTMBFS\ structure is collection of \FTBFS\ structures with respect to a collection of sources $S \subseteq V$.
For every $n$-vertex graph $G=(V,E)$ and a source set $S \subseteq V$, \cite{ParterP16} presented an algorithm for computing an \FTMBFS\ subgraph $H \subseteq G$ with $O(\sqrt{|S|}n^{3/2})$ edges. This was also shown to be existentially tight. Parter \cite{parter2015dual} presented a construction of dual-failure \FTBFS\ structures with $O(n^{5/3})$ edges. Gupta and Khan \cite{gupta2017multiple} extended this construction to multiple sources $S$ and provided a dual-failure \FTMBFS\ with $O(|S|^{1/3}n^{5/3})$ edges, which is also existentially tight \cite{ParterP16}.  For a general bound on the number of fault $f$, the state-of-the-art upper bound is $O(|S|^{1/2^f}n^{2-1/2^{f}})$ by Bodwin et al. \cite{bodwin2017preserving}, a lower bound of $\Omega(|S|^{1/(f+1)}n^{2-1/(f+1)})$ is known by \cite{parter2015dual}. Closing this gap is a major open problem. 

Fault-tolerant (FT) additive spanners are sparse subgraphs that preserve distance under failure with some additive stretch. While various upper bound constructions are known \cite{BraunschvigCPS15,bilo2015improved,parter2017vertex}, to this date no lower bounds are known for constant additive stretch. 
For example, one can compute $+2$ FT-additive spanners with $\widetilde{O}(n^{5/3})$ edges\footnote{The notation $\widetilde{O}$ hides poly-logarithmic terms in the number of vertices $n$.}, but no lower-bound of $n^{1/2+\epsilon}$ edges, for any $\epsilon>0$ is known.
%

\paragraph{Distributed Constructions.} Despite the fact that the key motivation for fault tolerant preservers comes from distributed networks, considerably less is known on their distributed constructions.
In this paper, we consider the standard $\mathsf{CONGEST}$  model of distributed computing \cite{Peleg:2000}. In this model, the network is abstracted as an $n$-node graph $G=(V, E)$, with one processor on each node. Initially, these processors only know their incident edges in the graph, and the algorithm proceeds in synchronous communication rounds over the graph $G=(V,E)$. In each round, nodes are allowed to exchange $O(\log n)$ bits with their neighbors and perform local computation. Throughout, the diameter of the graph $G=(V,E)$ is denoted by $D$. 

Ghaffari and Parter \cite{GhaffariP16} presented the first distributed constructions of fault tolerant distance preserving structures. For every $n$-vertex $D$-diameter graph $G=(V,E)$ and a source vertex $s \in V$, they gave an $\widetilde{O}(D)$-round algorithm for computing  an \FTBFS\ 
structure with respect to $s$. Both the size bound of the output structure and the round complexity of their algorithm are nearly optimal. An additional useful property of that algorithm is that it also computes the length of all the $\{s\}\times V$ replacement paths in the graph $G \setminus \{e\}$ for every $e \in G$. To the best of our knowledge, currently there are no non-trivial distributed constructions that support either multiple sources or more than a single fault. A natural extension of \cite{GhaffariP16} to a subset of sources $S$ (resp., to dual faults) might lead to a round complexity of $\Omega(|S| D)$ (resp., $\Omega(D^2)$ rounds). These bounds are inefficient for graphs with a large diameter, or for supporting a large number of sources.
Finally, while distributed constructions for additive spanners are known in the fault-free setting \cite{Pettie-Skeleton,censor2018distributed,ElkinM19}, there are no distributed constructions for the fault-tolerant setting. 
%
%
%
%

\subsection{Our Results}
We present constructions of 
FT preservers and additive spanners resilient to two edge failures with \emph{sublinear} round complexities. Throughout, we consider unweighted undirected $n$-vertex graph $G=(V,E)$ of diameter $D$.
 \\ \\
\paragraph{Fault Tolerant Distance Preservers.}
Given an unweighted and undirected $n$-vertex graph $G=(V,E)$ and integer $f \geq 1$, a subgraph $H \subseteq G$ is an $f$-\FTMBFS\ structure w.r.t. $S$ if:
$$\dist(s,t, H \setminus F)=\dist(s,t, G \setminus F), \mbox{~for every~} s \in S, t \in V, F \subseteq E \mbox{~and~} |F|\leq f~.$$ When $f=1$, we call $H$ an \FTMBFS\ structure, and when $f=2$ it is called a dual-failure \FTMBFS.

\begin{theorem}[Distributed FT-MBFS]\label{thm:FT-MBFS-dist}
There exists a randomized algorithm that given an $n$-vertex graph $G=(V,E)$, and a subset $S \subseteq V$ computes w.h.p. a subgraph $H \subseteq G$ such that $H$ is an \FTMBFS\ w.r.t. $S$ and $|E(H)|=O(\sqrt{|S|}n^{3/2})$ edges. The round complexity is $\widetilde{O}(D+\sqrt{n |S|})$.
\end{theorem}
This improves upon the $O(|S|D)$ bound implied by the FT-BFS algorithm of \cite{GhaffariP16}, for $D\geq \sqrt{n/|S|}$. 
The size of the \FTMBFS\ subgraph $H$ is existentially optimal (up to a logarithmic factor). 

We also consider the dual-failure setting. In the centralized literature it has been widely noted that the dual-failure case is already considerably more involved compared to the single fault setting. Indeed, there has been no prior distributed constructions of distance preserving subgraphs that are resilient to two faults. 
We provide a simplified centralized algorithm for the dual failure setting which serves the basis for our distributed construction:
\begin{theorem}[Distributed Dual-Failure FT-MBFS]\label{thm:FT-MBFS-dist-dual}
There exists a randomized algorithm that given an $n$-vertex graph $G=(V,E)$, and a subset $S \subseteq V$ computes w.h.p. a subgraph $H \subseteq G$ such that $H$ is a dual-failure \FTMBFS\ w.r.t. $S$ and contains 
$O(|S|^{1/8}\cdot n^{15/8})$ edges. The round complexity is $\widetilde{O}(D+n^{7/8}|S|^{1/8}+|S|^{5/4}n^{3/4})$.
\end{theorem}
We note that the size of our subgraph is suboptimal, as there exist (centralized) constructions \cite{gupta2017multiple,parter2015dual} that compute dual-failure \FTMBFS\ subgraphs with $O(|S|^{1/3}n^{5/3})$ edges. These constructions, however, are inherently sequential, and it is unclear how to efficiently implement them in the distributed setting. 
Specifically, in the \congest\ model, a naive simultaneous computation of multiple BFS trees $\BFS(s, G \setminus \{e_1,e_2\})$ for every $s \in S$ and $e_1,e_2 \in G$  might result in a very large \emph{congestion} over the graph edges. To reduce this congestion, one needs to balance the edge congestion and the sparsity of the output subgraph. These two opposing forces lead to suboptimal constructions w.r.t. the size, but with the benefit of obtaining a sub-linear round-complexity.
We also note that our algorithms solve the subgraph problem rather than the distance computation problem. That is, in contrast to the \FTBFS\ algorithm of \cite{GhaffariP16}, we compute only the FT preserving subgraph but not necessarily the FT distances.  
\\ \\
\paragraph{Fault Tolerant Additive Spanners.}
We employ the distributed construction of \FTMBFS\ structures to provide the first non-trivial constructions of fault tolerant $+2$ additive spanners. These structures are defined as follows. 
Given an unweighted undirected $n$-vertex graph $G=(V,E)$ and integer $f \geq 1$ and a stretch parameter $\beta$, a subgraph $H \subseteq G$ is an $+\beta$ $f$-FT additive spanner if:
$$\dist(s,t, H \setminus F)\leq \dist(s,t, G \setminus F)+\beta, \mbox{~for every~} s,t \in V, F \subseteq E \mbox{~and~} |F|\leq f~.$$
When $f=1$, we call $H$ an $+\beta$ FT-additive spanner, and when $f=2$ it is called a $+\beta$ dual-failure FT-additive spanner. 
By using Thm. \ref{thm:FT-MBFS-dist} and Thm. \ref{thm:FT-MBFS-dist-dual} respectively, we get:
\begin{corollary}[+2 Additive Spanner, Single Fault]\label{cor:two-add-onef}
For every $n$-vertex graph $G=(V,E)$, there exists a randomized algorithm that w.h.p. computes $+2$ FT-additive spanner $H \subseteq G$ with $\widetilde{O}(n^{5/3})$ edges in $\widetilde{O}(D+n^{5/6})$ rounds.
\end{corollary}
The size of the $+2$ FT-additive spanner matches the state-of-the-art bound. For the dual-failure setting, our bounds are suboptimal due to the suboptimality of the dual-failure \FTMBFS\ structures.
\begin{corollary}[+2 Additive Spanner, Two Faults]\label{cor:two-add-twof}
For every $n$-vertex graph $G=(V,E)$, there exists a randomized algorithm that w.h.p. computes a $+2$ dual-failure FT-additive spanner with $\widetilde{O}(n^{17/9})$ edges within $\widetilde{O}(D+n^{8/9})$ rounds. 
\end{corollary}
No sublinear round algorithms for +2 FT-additive spanners with $o(n^2)$-edges were known before.

\paragraph{The High-Level Approach.} We provide the high-level ideas required to compute 
\FTMBFS\ structures w.r.t. a collection of source nodes $S$. This construction is later on used for computing FT-additive spanners and dual-failure \FTMBFS\ structures. By definition, an \FTMBFS\ structure for $S$ is required to contain a BFS tree w.r.t. each $s \in S$ in the graph $G \setminus \{e\}$ for every $e \in G$. Upon using any consistent tie-breaking of shortest-path distances, the union of all these trees contain $O(|S|n^{3/2})$ edges \cite{ParterP16}.
Our goal is then to compute all these BFS trees efficiently in the \congest\ model.
For the single source case, the key observation made in \cite{GhaffariP16} is that for every vertex $t$, it is sufficient to send only $O(D)$ BFS tokens throughout the computation: one token for every edge $e$ on the shortest $s$-$t$ path $\pi(s,t)$. The reason is that a failing of an edge $e \notin \pi(s,t)$ does not effect the $s$-$t$ distance. Since every edge $(u,t)$ is required to pass through only $|\pi(s,t)|=O(D)$ BFS tokens, using the random-delay approach, all these trees could be computed in dilation+congestion$=\widetilde{O}(D)$ rounds, w.h.p. Extending this idea to multiple sources, ultimately leads to a round complexity of $\Omega(|S|D)$. Indeed a-priori it is unclear how to break this barrier, as for every $s$ and $e \in \pi(s,t)$, the $s$-$t$ distance in $G \setminus \{e\}$ might be different, forcing $t$ to receive the BFS token from $\Omega(|S|D)$ BFS algorithms. 

The key idea is to define for every vertex $t$ a smaller set of \emph{relevant pairs} $(s,e)$ from which it is allowed to receive the BFS tokens. This set $\{(s,e)\}$ is defined by including only edges $e$ that are sufficiently \emph{close} to $t$ on its $\pi(s,t)$ path for every $s$. The main technical issue that arises with this idea is the inconsistency in the definition of relevant pairs between nodes on a given replacement path $P(s,t,e)$. In particular, there might be cases where $(s,e)$ is relevant for $t$ but it is not in the relevant set of some vertex $w$ on the $P(s,t,e)$ path. In such a case, $w$ might block the propagation the BFS token $BFS(s,G \setminus \{e\})$ (as $(s,e)$ is not in its relevant set) which would prevent $t$ from receiving it.
These technical issues become more severe in the dual failure setting, due to a more delicate interaction between the dual-failure replacement paths. In the very high level to mitigate this technicality, we add to the output structure a collection of an (FT-) BFS trees w.r.t. a \emph{randomly} sampled set of nodes. This edge set would compensate (in a non-trivial manner) for the lost tokens of the truncated BFS constructions.

\subsection{Preliminaries}
Given an unweighted $n$-vertex graph $G=(V,E)$, for an $s,t \in V$ and $e \in G$, the replacement path $P(s,t,e)$ is an $s$-$t$ shortest path in $G\setminus \{e\}$. Throughout, we assume that the shortest-path ties are broken in a consistent manner using the vertex IDs. That is, in our BFS computations, the shortest path ties are broken by always preferring vertices of lower IDs.

For a given $p \in [0,1]$, let $\Sample(V,p)$ be a subset of vertices obtained by sampling each vertex in $V$ independently with probability $p$. Let $\BFS(s,G)$ be a BFS tree rooted at $s$ in $G$. 

The (unique) shortest-path between any pair $x,y$ in $G$ is denoted by $\pi(x,y,G)$, when the graph $G$ is clear from the context, we may omit it and write $\pi(x,y)$. Let $N(u)$ be the neighbors of $u$ in $G$.
Given a tree $T$ and $u,v \in V(T)$, let $\pi(u,v,T)$ be the tree path between $u$ and $v$.  For a given integer parameter $\sigma$, let $\pi_{\sigma}(u,v)$ denote the set of last $\min\{\sigma, |\pi(u,v)|\}$ edges (closest to $v$) on the path $\pi(u,v)$. For an edge $e=(x,y)$ and a subgraph $G'\subseteq G$, let $\dist(e,t,G')=\min\{\dist(x,t,G'), \dist(y,t,G')\}$. For an $s$-$t$ path $P$, let $\LastE(P)$ be the last edge of the path (incident to $t$). For $a,b \in P$, let $P[a,b]$ be the sub-path segment between $a$ and $b$ in $P$. 
\begin{definition}\label{def:ft-trees}
For a given source vertex $s$, a subgraph $H$ is an \emph{FT-BFS structure} with respect to $s$ if 
$\dist(s,t,H \setminus \{e\})=\dist(s,t,G \setminus \{e\})$ for every $t \in V$ and $e \in E$. In the same manner, for a given subset $S \subseteq V$, a subgraph $H$ is an \emph{multi-source FT-BFS structure} with respect to $S$ if 
$\dist(s,t,H \setminus \{e\})=\dist(s,t,G \setminus \{e\})$ for every $s,t \in S \times V$ and $e \in E$.
\end{definition}

\begin{fact}\cite{ParterP16}\label{obs:last-edge}
For every $n$-vertex graph $G=(V,E)$, and a subset $S \subseteq V$, let 
$$H=\bigcup_{s,t,e \in S\times V \times E}\{\LastE(P(s,t,e))\}.$$
Then, $H$ is an \FTMBFS\ structure w.r.t. $S$ and $|E(H)|=O(\sqrt{|S|}\cdot n^{3/2})$. This edge bound is tight.
\end{fact}

\paragraph{The random delay technique.}
Throughout, we make an extensive use of the random delay approach of \cite{leighton1994packet,Ghaffari15}. Specifically, we use the following theorem:
\begin{theorem}[{\cite[Theorem 1.3]{Ghaffari15}}]\label{thm:delay}
Let $G$ be a graph and let $A_1,\ldots,A_m$ be $m$ distributed algorithms in 
the $\mathsf{CONGEST}$ model,  where each algorithm takes at most $\dilation$ rounds, and where for each 
edge of $G$, at most $\congestion$ messages need to go through it, in total 
over all these algorithms. Then, there is a randomized distributed
algorithm (using only private randomness) that, with high probability, 
produces 
a schedule that runs all the algorithms in $O(\congestion +\dilation \cdot 
\log 
n)$ rounds, after $O(\dilation \log^2 n)$ rounds of pre-computation.
\end{theorem}

\section{~Simplified Meta-Algorithms} 
We start by presenting simplified (centralized) constructions of FT preserving subgraphs. These constructions serve as a more convenient starting point for the distributed constructions described in the next sections.

\paragraph{\FTMBFS\ Structures.}
Let $T_S=\bigcup_{s \in S}T_s$ where $T_s=\BFS(s,G)$. 
The \FTMBFS\ subgraph $H$ is given by the union of three subgraphs: $T_S$, and the subgraphs $H_1$ and $H_2$ defined by:
$$H_1=\{\LastE(P(s,t,e)) ~\mid~ s,t \in S \times V, e \in \pi_\sigma(s,t)\} \mbox{~where~} \sigma=\sqrt{n/|S|},$$
and $H_2=\bigcup_{r \in R} \BFS(r, G) \mbox{~~where~~} R=\Sample(V,10\log n/\sigma)~.$
\begin{lemma}
$E(H)=O(\sqrt{|S|}\cdot n^{3/2}\log n)$ and $H$ is an \FTMBFS\ with respect to $S$.
\end{lemma}
\begin{proof}
The size analysis follows by noting that $|T_S|=O(|S|\cdot n)$, and in addition each vertex $t$ adds at most $\sigma=\sqrt{n/|S|}$ edges to $H_1$ for every source $s \in S$. Thus, $|H_1|=O(\sigma \cdot |S|n)=O(n\sqrt{n |S|})$.
Turning to $H_2$, by the Chernoff bound, w.h.p., $|R|=O(n\log n/\sigma)$ and thus $|H_2|=O(\sqrt{|S|}\cdot n^{3/2}\log n)$.

We next show that $H$ is an \FTMBFS\ with respect to $S$. By Fact \ref{obs:last-edge}, it is sufficient to show that $H$ contains the last edge of the replacement path $P(s,t,e)$ for every $s,t,e\in S \times V \times E$.
Fix a source $s \in S$ and a vertex $t\in V$. If $\dist(s,t,G)\leq \sigma$, then $H_1 \cup T_S$ contain the last edge of $P(s,t,e)$ for every edge $e$. This is because $\LastE(P(s,t,e))$ is added to $H_1$ for every $e \in \pi(s,t,T_s)$ and $P(s,t,e)=\pi(s,t,T_s)$ for every $e \notin \pi(s,t,T_s)$. 

Thus, assume that $\dist(s,t,G)\geq \sigma$ and specifically, consider an edge $e\in \pi(s,t,T_s)\setminus \pi_{\sigma}(s,t,T_s)$. Since $\dist(s,t,G)\geq \sigma$, it also holds that $|P(s,t,e)|\geq \sigma$. Thus by the Chernoff bound, w.h.p. there is at least one vertex in $R$ that lies in the $(\sigma/2)$-length suffix of $P(s,t,e)$. That is, w.h.p., there is a vertex $r \in V( P(s,t,e))\cap R$ such that $\dist(r,t, P(s,t,e))\leq \sigma/2$. We next claim that there is \emph{no} $r$-$t$ shortest path in $G$ that contains the failing edge $e$. This holds as $\dist(r,t,G)\leq \dist(r,t,G \setminus \{e\})\leq \sigma/2$, but by the definition of the edge $e$, 
$\dist(e,t,G)\geq \sigma$. By the uniqueness of the replacement paths, we have that $P(s,t,e)[r,t]=\pi(r,t,T_r)$ where $T_r=\BFS(r,G)$, and thus $\LastE(P(s,t,e)) \in H_2$. The claim follows. 
\end{proof}

\paragraph{Dual-Failure \FTMBFS\ Structures.} We next describe a simplified centralized construction of dual-failure 
\FTMBFS\ structures, this serves the basis for the distributed implementation. As we will see later on, computing these structures in the distributed setting is considerably more involved. To balance between edge congestion and sparsity of the structure, the final size of the \FTMBFS\ structures computed in distributed setting is larger compared to the centralized setting. For every $s, t \in V$ and $e_1,e_2 \in E$, let $P(s,t,\{e_1,e_2\})$ be the $s$-$t$ shortest path in $G \setminus \{e_1,e_2\}$.
\begin{fact}\cite{parter2015dual}\label{obs:dual-S}
For every $n$-vertex graph $G=(V,E)$, and a subset $S \subseteq V$, let 
$$H=\bigcup_{s,t \in S\times V, e_1,e_2 \in E}\{\LastE(P(s,t,\{e_1,e_2\}))\}.$$
Then $H$ is a dual-failure \FTMBFS\ w.r.t. $S$. 
\end{fact}
Let $S$ be the set of sources. Let $R_1$ be a random sample of $O(\sqrt{n|S|}\log n)$ vertices, and let $R_2$ be a random sample of $O(|S|^{1/4}\cdot n^{3/4}\log n)$ vertices. 
Let $H_1=\bigcup_{r \in R_1} \FTMBFS(r,G)$ and $H_2=\bigcup_{r \in  R_2} \BFS(r,G)$. The dual-failure $\FTBFS$ structure w.r.t. $S$ denoted by $H$ contains the subgraphs $T_S$, $H_1$, $H_2$ as well as the a subset of last edges of certain replacement paths. Let $\sigma_1=\sqrt{n/|S|}$ and $\sigma_2=(n/|S|)^{1/4}$. For every path $P$ and integer $\sigma$, let $P_{\sigma}$ be the $\sigma$-length suffix of $P$ (when $\sigma\geq |P|$, then $P_{\sigma}$ is simply $P$). Every vertex $t$, define the edge set $E_t$ as
$$E_t= \bigcup_{s \in S}\bigcup_{e_1 \in \pi_{\sigma_1}(s,t)} \bigcup_{e_2 \in P_{\sigma_2}(s,t,e_1)}\{\LastE(P(s,t,\{e_1,e_2\})\}~.$$
The final dual-failure $\FTBFS$ structure is given by:
$$H= H_1 \cup H_2 \cup \bigcup_t E_t~.$$
\begin{lemma}\label{lem:simple-dual}
W.h.p., $H$ is a dual-failure \FTMBFS\ w.r.t. $S$, and $|E(H)|=\widetilde{O}(|S|^{1/4}\cdot n^{7/4})$ edges.
\end{lemma}

\begin{proof}
By the definition of the $E_t$ sets, it remains to show that $H$ contains the last edge of an replacement path $P(s,t,\{e_1,e_2\})$ such that either (i) $e_1 \in \pi(s,t) \setminus \pi_{\sigma_1}(s,t)$ or (ii) $e_1 \in \pi_{\sigma_1}(s,t)$ but $e_2 \in P(s,t,e_1) \setminus P_{\sigma_2}(s,t,e_1)$. We begin with (i). Since $e_1 \in \pi(s,t) \setminus \pi_{\sigma_1}(s,t)$, we have that $|P(s,t,e_1)|\geq \sqrt{n/|S|}$. Since we sample each vertex into $R_1$ with probability of $10\log n\cdot \sqrt{|S|/n}$, w.h.p. the $\sigma_1/2$-length suffix of $P(s,t,\{e_1,e_2\})$ contains a vertex, say $r$, in $R_1$. Since $\dist(r,t,G)\leq \sigma_1/2$, we have that $e_1 \notin \pi(r,t,G)$, and by the uniqueness of the shortest paths, we have that $P(s,t,\{e_1,e_2\})=P(s,t,\{e_1,e_2\})[s,r]\circ P(r,t,\{e_2\})$. Since $H_2$ contains the FT-BFS w.r.t. $r$, 
it contains the path $P(r,t,\{e_2\})$ and thus $\LastE(P(s,t,\{e_1,e_2\}))$ is in $H$.
We proceed with (ii). Since $e_2 \in P(s,t,e_1) \setminus P_{\sigma_2}(s,t,e_1)$, we have that $|P(s,t,\{e_1,e_2\})|\geq \sigma_2$. Since we sample each vertex into $R_2$ with probability of $10\log n/\sigma_2$, w.h.p. the $\sigma_2/2$-length suffix of $P(s,t,\{e_1,e_2\})$ contains a vertex, say $r'$, in $R_2$. Since $\dist(r,t,G)\leq \sigma_2/2$, we have that $e_1,e_2 \notin \pi(r,t,G)$. By the uniqueness of the shortest paths, we have that 
$P(s,t,\{e_1,e_2\})=P(s,t,\{e_1,e_2\})[s,r]\circ \pi(r,t,G)$. The claim follows as $H_1$ contains the BFS tree rooted at $r$, and concluding that $\LastE(P(s,t,\{e_1,e_2\}))$ is in $H$.
The size bound follows by noting that $|E(H_1)|=O(|\sqrt{|R_1|}n^{3/2})$ and $|E(H_2)|=O(|R_2|n)$. In addition, since each vertex $t$ adds the last edges of $O(|S|\cdot (n/|S|)^{3/4})$ replacement paths, we get that $|E_t|=O(|S|^{1/4}n^{3/4})$. The lemma follows.
\end{proof}

\paragraph{Comparison to Bodwin et al. \cite{bodwin2017preserving}.} A simplified algorithm for computing sparse \FTMBFS\ structures (of suboptimal size) has been also provided by \cite{bodwin2017preserving}. Their algorithm iterates over the vertices where for every vertex $t$ the algorithm defines a small set of edges incident to $t$ that should be added to the output subgraph $H$. For every \emph{vertex} $t$, the algorithm reduces the task of computing \FTMBFS\ structure with respect to $S$ sources and supporting $f$ faults\footnote{The task is to pick the edges of $t$ that should be added to such a structure to provide resilience against $f$ faults.} into the  computation of an \FTMBFS\ structure to support $S'$ sources and $f-1$ faults, where $|S'|=O(\sqrt{|S|n})$. The main limitation in implementing this algorithm in the distributed setting is that for each vertex $t$ the algorithm defines a \emph{distinct} set of sources. For $f=1$ for example, our simplified algorithm computes BFS trees w.r.t. a subset of sources $S'$. In contrast, in the algorithm of \cite{bodwin2017preserving}, a BFS tree is computed w.r.t. a distinct set of sources $S_t$ for every vertex $t$, the union of all these $S_t$ sets might be very large (leading to a large round complexity).

\section{Distributed Construction of $\FTMBFS$ Structures}
In this section we prove Thm. \ref{thm:FT-MBFS-dist} and present our main algorithm for computing sparse $\FTMBFS$ structures with respect to $S$ sources. This structure becomes useful both for the construction of FT-additive spanners, and for the computation of the dual-failure FT preservers. 

\subsection{The algorithm}
Set $\sigma=\lceil \sqrt{n/|S|} \rceil$ and $\sigma'=3\sigma$. The algorithm has two main steps. In the first step, a subset $R$ of $O(n\log n/\sigma)$ vertices is uniformly sampled, and a BFS trees $T_s=\BFS(s,G)$ is computed for every vertex $s \in S \cup R$. Let $T_S=\bigcup_{s \in S}T_s$ and  $T_R=\bigcup_{r \in R}T_r$. All the edges of $T_S \cup T_R$ are added to the output subgraph $H$, by their corresponding endpoints.

In the second step, the algorithm computes a special subset of replacement paths, and the last edges of these replacement paths are added to $H$. To define this subset, we need the following definition.
For every $s,t \in S \times V$, each vertex $t$ defines a set of \emph{relevant edge-list} $\pi_{\sigma'}(s,t)$ that consists of the last $\sigma'$ edges of its $\pi(s,t)$ paths. It also defines a shorter prefix $\pi_{\sigma}(s,t)$ that contains the last $\sigma$ edges of this path. 

The algorithm first lets each vertex $t$ learn its relevant edge-list $\Rel(s,t)$ for every $s \in S$. This can be done within $O(|S|\cdot\sigma'+D)=O(\sqrt{n |S|}+D)$ rounds by applying a simple pipeline strategy. 
From now on, the algorithm divides the time into phases of $\ell=O(\log n)$ rounds.
Every $\BFS(s,G\setminus \{e\})$ algorithm then starts in phase $\tau_{s,e}$, where $\tau_{s,e}$ is a random variable with a uniform distribution in
$[1, \sigma' \cdot |S|]$. Specifically, using the notion of $k$-wise independence, similarly to \cite{GhaffariP16}, all vertices can learn a random seed of $\mathcal{SR}$ of $O(\log n)$ bits. Using the seed $\mathcal{SR}$ and the IDs of the edge $e$ and the source $s$, all vertices can compute the starting phase $\tau_{s,e}$ of each BFS construction $\BFS(s, G \setminus \{e\})$. In the analysis section, we show that due to these random starting points, w.h.p., each edge $e'=(u,v)$ is required to send as most $\ell$ edges in every phase.  In a standard application of a BFS computation with delay $\tau_{s,e}$, every vertex $t$ is supposed to receive a $\BFS(s,G\setminus \{e\})$-token (for the first time) in phase $\dist(s,t,G \setminus \{e\})+\tau_{s,e}$. In our case, the algorithm cannot afford to compute the entire BFS trees, but rather only certain fragments of them.
Specifically, the BFS tokens are initiated and propagated following certain rules whose goal is to keep the congestion over the edges small. In the high-level, every vertex $t$ would send its neighbor $u \in N(t)$ (such that $(u,t)\neq e$) a BFS token $\BFS(s,G\setminus \{e\})$ only if $e \in \Rel(s,u)$. In the special case where $e \notin \pi(s,t)$, the vertex $t$ will initiate the BFS token to $u$ in phase $\dist(s,t,G \setminus \{e\})+\tau_{s,e}$. In the remaining case where $e \in \pi(s,t)$, $t$ will send $u$ the token $\BFS(s,G\setminus \{e\})$ to $u$ in phase $i$ iff (i) $e \in \Rel(s,u)$ and $t$ received the token $\BFS(s,G\setminus \{e\})$ for the first time in phase $i-1$. 

As we will see in the analysis section, even-though each vertex $t$ sends the BFS tokens $\BFS(s,G\setminus \{e\})$ to neighbors $u$ provided that $e \in \Rel(s,u)$, it might be the case that a vertex $u$ would not get the BFS token for each of its edges in $\Rel(s,u)$. This might happen when the path $P(s,u,e)$ contains intermediate vertices $w$ for which $e \notin \Rel(s,w)$, which would block the propagation of the token. Fortunately, a more careful look reveals that in all the cases where $\LastE(P(s,u,e))\notin T_R$, the BFS token of $\BFS(s,G\setminus \{e\})$ would complete its propagation over the entire $P(s,u,e)$ for every $e \in \SubRel(s,u) \subseteq \Rel(s,u)$. 
As we will see, this would be sufficient for the correctness of the \FTMBFS\ structure. 
\\
\\
\noindent\hrulefill
\noindent\textbf{The Distributed FT-MBFS Algorithm: }
\begin{enumerate}
\item Set $R \gets \Sample(V,10\log n/\sigma)$.
\item Compute BFS trees $T_u=\BFS(s,G)$ for every $s \in S \cup R$, and add these trees to $H$.
\item Number the edges of $T_S=\bigcup_{s \in S} T_s$ by numbers $1$ to $|S|(n-1)$, where each edge $e\in T_s$ has a distinct number for every $T_s$ containing $e$.
\item Make each vertex $v$ know the numbers of the edges on the $\sigma'$-length suffix $\Rel(s,v)$ for every $s \in S$.   
\item Let each vertex $v$ sends to each of its neighbors $u \in N(v)$ the numbers of the edges on $\bigcup_{s \in S}\Rel(s,v)$. 
\item Broadcast a string $\mathcal{SR}$ of $O(\log n)$ random bits.
\item For every $s \in S$, and $e \in T_s$, let $\tau_{s,e}$ be picked uniformly at random from $\{1, 2, \dots, 2\sigma' \cdot |S|\}$ by setting it equal to $\mathcal{SR}[i]$ where $i$ is the edge-number of $e$ in $T_s$. Since $\mathcal{SR}$ is publicly known, given the ID of an edge $e$ and the source $s$, a vertex can compute $\tau_{s,e}$.
\item Divide time into phases of $\ell=\Theta(\log n)$ rounds each.
\item Run each $\BFS(s,G\setminus \{e\})$ for every $e$ and $s \in S$ at a speed of one hop per phase, following these rules for every vertex $v$:
\begin{itemize}
\item For every neighbor $u \in N(v)$ and every edge $e \in \Rel(s,u) \setminus \pi(s,v)$\footnote{In Lemma \ref{lem:know-path}, we explain how $v$ can locally detect the edges in $\Rel(s,u) \setminus \pi(s,v)$.} , $v$ sends $u$ the BFS token $\BFS(s,G\setminus \{e\})$ in phase $\dist(s,v, G)+\tau_{s,e}$.
\item For every BFS token $BFS(s,G\setminus \{e\})$ received for the first time at phase $i$ at $v$ from a non-empty subset of neighbors $N'(v) \subseteq N(v)$, $v$ does the following:
\begin{itemize}
\item If $e \in \SubRel(s,v)$, then $v$ adds the edge $(w,v)$ to $H$ where $w$ is the vertex of minimum-ID in $N'(u)$. 
\item $v$ sends the BFS token $\BFS(s,G\setminus \{e\})$ in phase $i+1$ to every neighbor $u\in N(v) \setminus N'(v)$ satisfying that $e \in \Rel(s,u)$. 
\end{itemize}
\end{itemize}
 
\end{enumerate}
\hrulefill
\bigskip
\\
\paragraph{Second-Order Implementation Details.}
For the generation of the shared random seed, we use the same construction of \cite{GhaffariP16} which is based on the notion of $k$-wise independence hash functions. 
\begin{lemma}\label{lem:randomness}\cite{GhaffariP16} The string of shared randomness $SR$ can be generated and delivered to all vertices in $O(D + \log n)$ rounds.
\end{lemma}
We argue that each vertex $v$ by knowing the sets $\bigcup_{u \in \{v\}\cup  N(v)}\Rel(s,u)$, can locally compute the edges in $\Rel(s,u)\setminus \pi(s,v)$ for every $u \in N(v)$.  
\begin{lemma}\label{lem:know-path}
For every vertex $v$, neighbor $u \in N(v)$ and an edge $e \in \Rel(s,u)$, $v$ can locally decide if $e \in \pi(s,v)$ or not.
\end{lemma}
\begin{proof}
Let $e \in \Rel(s,u) \cap \pi(s,v)$. We will show that $e$ can locally recognize that $e \in \pi(s,v)$. 
If $e \in \Rel(s,v)$, then $v$ clearly knows that $e \in \pi(s,v)$. 
Otherwise, if $e=(x,y) \in \Rel(s,u)\setminus \Rel(s,v)$, we show that $y$ must be the endpoint of the first edge in $\Rel(s,v)$. To see this, assume towards contradiction that $y$ has no incident edge in $\Rel(s,v)$. Since $e \in \Rel(s,u)$, we have that $\dist(x,u,G)\leq 3\sigma$. However, by the assumption, $\dist(x,v,G)\geq 3\sigma+2$, in contradiction as $(u,v)$ are neighbors. 
As $y \in V(\Rel(s,v))$ and $e=(x,y) \in \Rel(s,u)$, $v$ can deduce that $x$ is the parent of $y$ in $T_s$, and consequently that $(x,y) \in \pi(s,v)$ as well. 
\end{proof}

Note that vertex $u$ receives messages from all its potential parents in $\BFS(s, G \setminus \{e\})$
at the same time, namely, at phase $\dist(s, u, G \setminus \{e\}) + \tau_{s,e}$. It
selects as its parent the vertex of minimum ID, which would guarantee that the shortest path ties are broken
an a consistent manner, leading to a sparse structure.
\\ \\
\paragraph{Correctness.} 
In the \FTBFS\ construction of \cite{GhaffariP16}, for every vertex $v$, the BFS token $\BFS(s, G \setminus \{e\})$ reached every vertex $t$ for which $e \in \pi(s,t)$. In contrast, in our setting, only a subset of the replacement paths are fully constructed which poses a challenge for showing the correctness.

To show that the output subgraph $H$ is indeed an $\FTMBFS$ w.r.t. $S$, throughout, we fix a source $s \in S$, target $t \in V$ and an edge $e=(x,y)$. We need the following definitions. 
Let $T_s$ be a BFS tree rooted at $s$ for every $s \in S$. 
For a vertex $y$ and a tree $T_s$, let $T_s(y)$ be the subtree rooted at $y$ in $T_s$. A vertex $w$ is said to be \emph{sensitive} to an edge $e \in T_s$, if $e \in \pi(s,w)$. Observe that for every edge $e=(x,y) \in T_s$, where $x$ is closer to $s$, the set of sensitive vertices to $e$ are those that belong to $T_s(y)$.

\begin{definition}[Sensitive-Detour]\label{def:sen-det-one-f}
For a given replacement path $P(s,t,e)$ let $w$ be the first vertex on the path (closest to $s$) that is sensitive to $e$. We denote the segment $SD(s,t,e)=P(s,t,e)[w,t]$ by the \emph{sensitive-detour} of $P(s,t,e)$.
\end{definition}

\begin{observation}
For every $s, t \in S \times V$ and $e=(x,y) \in G$, it holds that: (i) $SD(s,t,e) \subseteq T_s(y)$ and (ii) $P(s,t,e)=\pi(s,w') \circ (w',w) \circ SD(s,t,e)$ for a unique pair $w,w' \in P(s,t,e)$. 
\end{observation}
\begin{proof}
(i) Let $w$ be the first vertex in $T_s(y) \cap P(s,t,e)$, thus $SD(s,t,e)=P(s,t,e)[w,t]$. 
Assume towards contradiction that there exists $w' \in SD(s,t,e)$ such that $w' \notin T_s(y)$. Since the shortest-paths are computed in a consistent manner, and $e \notin \pi(s,w')$, we get that $P(s,t,e)[s,w']=\pi(s,w')$. Thus, $w \in \pi(s,w')$, contradiction as $w \in T_s(y)$. 

(ii) Let $w'$ be the neighbor of $w$ (defined as above) on $P(s,t,e)$ that is closer to $s$. By definition, $w' \notin T_{s}(y)$ and thus by the uniqueness of the shortest-paths, we have $P(s,t,e)=\pi(s,w') \circ (w',w) \circ SD(s,t,e)$.
\end{proof}

\begin{claim}\label{cl:every-sen}
If $e \in \pi_{\sigma'}(s,w')$ for every $w' \in SD(s,t,e)$, then $\LastE(P(s,t,e))\in H$.
\end{claim}
\begin{proof}
Let $w$ be the first vertex on $SD(s,t,e)$ and let $q$ be the preceding neighbor of $w$ (not in $SD(s,t,e)$).
Since $e \in \pi_{\sigma'}(s,w)$, the vertex $q$ can locally detect that $e \notin \pi(s,q)$ (using Lemma \ref{lem:know-path}). Note that since $q \notin SD(s,t,e)$, it holds $\dist(s, q, G \setminus \{e\})=\dist(s, q, G)$.
Thus, $q$ send to $w$ the BFS token $\BFS(s,G \setminus \{e\})$ in phase $\dist(s, q, G \setminus \{e\})+1+\tau_{s,e}$. The token propagates over the $SD(s,t,e)$ segment at a speed of one hop per phase as for each $w' \in SD(s,t,e)$, $e \in \pi_{\sigma'}(s,w')$. 
\end{proof}

\begin{lemma}\label{lem:main-cor}
For every $s,t \in S \times V$ and $e \in G$, we have that $\LastE(P(s,t,e))\in H$. 
\end{lemma}
\begin{proof}
Fix a replacement path $P(s,t,e)$ where $e=(x,y)$. 
We consider the following cases.

\paragraph{Case (1): $e \in \pi(s,t)\setminus \SubRel(s,t)$.} In this case, $|P(s,t,e)| \geq \dist(s,t,G) \geq \sigma$ and thus w.h.p. the $\sigma/2$-length suffix of the path contains at least one sampled vertex in $R$, say $r$. 
Since $\dist(r,t,G)\leq \sigma/2$ but $\dist(x,t,G)\geq \sigma$, the edge $e$ does not appear on any $r$-$t$ shortest path. As the shortest-path ties are broken in a consistent manner, we have that $P(s,t,e)[r,t]=\pi(r,t)$.
Since the algorithm adds the BFS trees w.r.t. all vertices in $R$, we have that $\LastE(\pi(r,t))\in H$.
\\
\paragraph{Case (2): $e\in \SubRel(s,t)$ but $|SD(s,t,e)|\geq \sigma$.}
The proof for this case follows by noting that for every two vertices $u,v \in T_s(y)$, there is no $u$-$v$ shortest path that go through the edge $e$. Assume towards contradiction that there is a $u$-$v$ shortest path $P$ that goes through $e$, since $\pi(x,u)\subset \pi(s,u)$, $\pi(x,v)\subset \pi(s,v)$, it holds that $e \in \pi(x,u), \pi(x,v)$, and thus:
$$|P|=\dist(u,x,G)+\dist(x,v,G)=1+\dist(y,u,G)+1+\dist(y,v,G)=2+\dist(u,v,G)~,$$
contradiction that $P$ is a $u$-$v$ shortest path. Since $|SD(s,t,e)|\geq \sigma$, w.h.p., it contains at least one sampled vertex $r \in R$. As both $r,t \in T_s(y)$, $\pi(r,t)$ is free of failed edge $e$. Thus $P(s,t,e)[r,t]=\pi(r,t)$, concluding that $\LastE(P(s,t,e))\in H$. We note that this is the only case where the proof would not work for the case of a single vertex (rather than edge) fault.
\\
\paragraph{Case (3): $e\in \SubRel(s,t)$ but $|SD(s,t,e)|< \sigma$.}
This is the most interesting case as the last edge of the path $P(s,t,e)$ is not necessarily in $\bigcup_{r \in R}T_r$. We need to show that the suffix of the path $P(s,t,e)$ is computed by the algorithm, and that its last edge is added to $H$. Since $|SD(s,t,e)|< \sigma$, it holds that $\dist(w,t, G \setminus \{e\})\leq \sigma$ where $w$ is the first vertex on the $SD(s,t,e)$ segment. Since $e \in \pi_{\sigma}(s,t)$, it holds that $\dist(e,w, G)\leq 2\sigma$. Finally, as $w \in  SD(s,t,e)$ it implies that $e \in \pi_{\sigma'}(s,w)$. The claim then follows by Claim \ref{cl:every-sen}.
\end{proof}

\paragraph{Size.} The first part adds the BFS trees w.r.t. $|R|=O(\sqrt{|S|n})$ vertices. In addition, in the second step of the truncated BFS constructions, for every edge $e \in \bigcup_{s \in S}\SubRel(s,v)$, the vertex $v$ adds at most one edge to $H$ (corresponding to the last edge of $P(s,v,e)$). Since $|\bigcup_{s \in S}\SubRel(s,v)|=O(\sqrt{|S| n})$, this adds $O(\sqrt{|S|}n^{3/2})$ edges.

\paragraph{Round Complexity.}
\begin{claim}\label{cl:fast-pre}
Each vertex $t$ can learn the relevant edge set $\bigcup_{s \in S}\Rel(s,t)$ within $O(\sqrt{n |S|}+D)$ rounds.
\end{claim}
\begin{proof}
For every $s \in  S$, each edge $e$ in $T_s$ propagates down the tree for $\sigma'$ time steps (i.e., until reaching all vertices at distance $\sigma'$ from $e$). Focusing on a single-source $s$, each edge $e'=(x,y)$ needs to pass at most $\sigma'$ messages, corresponding to the last $\sigma'$ edges on the $\pi(s,y)$ path. Since there are $|S|$ sources, the total number of messages passing through a single edge is $|S|\cdot\sigma'$. Using pipeline all these messages can arrive in $O(\sqrt{n |S|}+D)$ rounds. 
\end{proof}
\begin{lemma}\label{lem:many-BFS}
W.h.p., at most $\ell=O(\log n)$ BFS tokens need to go through each edge, per phase.
\end{lemma}
\begin{proof} 
We show that w.h.p., in each phase number $\tau$ and for each edge $e'=(v, u)$, at most $O(\log n)$ BFS tokens will need to go through $e'$ from $v$ to $u$ in phase $\tau$. 
Note that the only BFS tokens passing over the edge $e'=(v,u)$ correspond to the BFS algorithms of $\BFS(s,G\setminus \{e\})$  for $e \in \Rel(s,u) \cup \Rel(s,v)$. Thus each edge passes $O(|S|\cdot \sigma)$ tokens. 
 
Each of the permitted tokens $\BFS(s, G \setminus \{e\})$ passing through $e'$ from $v$ to $u$ in phase $\tau$ satisfies that $\dist(s,v, G\setminus \{e\}) + \tau_{s,e} = \tau$. 
Assuming that the starting phase $\tau_{s,e}$ is chosen uniformly at random from a range of size $6\sigma |S|$, the probability of that event is at most $1/(6\sigma|S|)$. Hence, over the set $6\sigma\cdot |S|$ permitted tokens, only $1$ token, in expectation, is scheduled to go through from $v$ to $u$ in phase $\tau$. If the random delay values $\tau_{s,e}$ were completely independent, by an application of the Chernoff bound, we would have that this number is at most $O(\log n)$, w.h.p. This will not be exactly true in our case, as we produce $\mathcal{SR}$ using a pseudo-random generators, but using $k$-wise independence on the generated string, for $k=\Theta(\log n)$ and from a result of Schmidt et al.\cite{schmidt1995chernoff}, it is known that for this application of the Chernoff bound, it suffices to have $k$-wise independence between the random values, for $k=\Theta(\log n)$.   
\end{proof}
We are now ready to complete the round complexity argument. By Claim \ref{cl:fast-pre}, each vertex $t$ computes its relevant edge set $\pi_{\sigma'}(s,t)$ within $O(D+\sqrt{|S|n})$ rounds. Within additional $O(\sqrt{|S|n})$ rounds, each vertex $t$ can also learn the relevant edge sets of its neighbors.
By Lemma \ref{lem:many-BFS}, the computation of all BFS trees is implemented within $\widetilde{O}(D+\sigma \cdot |S|)=\widetilde{O}(D+\sqrt{|S|n})$ rounds. This completes the proof of Theorem \ref{thm:FT-MBFS-dist}.

\section{~Distributed Construction of Dual Failure Distance Preservers}
In this section, we extend the construction of \FTMBFS\ structures to support two edge failures.
Throughout, for every $s,t \in S \times V$, and $e_1,e_2 \in G$, recall that $P(s,t,\{e_1,e_2\})$ is the unique $s$-$t$ path in $G \setminus \{e_1,e_2\}$ chosen based on a consistent tie-breaking scheme (based on vertex IDs). 
For a given parameter $\sigma$, let $P_{\sigma}(s,t,F)$ be the $\sigma$-length suffix (ending at $t$) of the path. When $|P(s,t,F)|\leq \sigma$, $P_{\sigma}(s,t,F)$ is simply $P(s,t,F)$. 
Set
$$\sigma_1=(n/|S|)^{5/8} \mbox{~~and~~} \sigma_2=(n/|S|)^{1/4}~.$$
We start by describing the algorithm based on the assumption that every vertex $t$ has the following information:
\begin{itemize}
\item{(I1)} The distance $\dist(s,t,G \setminus \{e\})$ for every $s \in S$ and every $e \in \pi_{2\sigma_2}(s,t)$.
\item{(I2)} The path segment $P_{\sigma_2}(s,t,e)$ for every $s \in S$ and every $e \in \pi_{2\sigma_2}(s,t)$.
\end{itemize}
Note that in contrast to the \FTBFS\ construction of \cite{GhaffariP16}, the \FTMBFS\ algorithm of Theorem \ref{thm:FT-MBFS-dist} computes the structure but not necessarily the distances. We therefore need to augment the algorithm by a procedure that computes the information (I1,I2) for all near faults (at distance at most $\sigma_2$ from $t$). 

\begin{lemma}\label{lem:short-dist}
There is a randomized algorithm that w.h.p. computes the information (I1,I2) for every vertex $t$ within $\widetilde{O}((n/\sigma_1) \cdot \sigma_2+ (n/\sigma_1)^2+D)$ rounds. 
\end{lemma}
In Subsec. \ref{sec:key-dual} we describe the key construction. Then in Subsec. \ref{sec:short-dist}, we prove Lemma \ref{lem:short-dist}. 

\subsection{~Distributed Alg. for Dual Failure $\FTMBFS$ Structure (Under the Assumption)}\label{sec:key-dual}
Before explaining the algorithm, we need the following definition, which extends Def. \ref{def:sen-det-one-f} to the dual failure setting. 
\begin{definition}[Sensitive-Detours of Dual-Fault Replacement Paths]\label{def:sen-det-two-f}
A vertex $t$ is \emph{sensitive} to the triplet $(s,e_1,e_2)$ if $P(s,t,\{e_1,e_2\})\notin \{P(s,t,e_1), P(s,t,e_2)\}$. This necessarily implies that for a sensitive vertex it holds that $e_2 \in P(s,t,e_1)$ \emph{and} $e_1 \in P(s,t,e_2)$. For a given $P(s,t,\{e_1,e_2\})$ path, let $w$ be the first vertex (closest to $s$) that is sensitive to $(s,e_1,e_2)$. The sensitive-detour $SD(s,t,\{e_1,e_2\})$ correspond to the segment $P(s,t,\{e_1,e_2\})[w,t]$.
\end{definition}
Set $\sigma=\sigma_2=(n/|S|)^{1/4}$. The first step of the algorithm computes an \FTMBFS\ subgraph $\FTMBFS(R \cup S)$ where $R$ is a randomly sampled set of $O(n\log n/\sigma)$ vertices. By Thm. \ref{thm:FT-MBFS-dist}, this can be done in $\widetilde{O}(\sqrt{|R|n}+D)$ rounds. 
The second step computes a subset of dual-failure replacement paths $\{P(s,t,\{e_1,e_2\}), s \in S, t \in V, e_1,e_2 \in E\}$ that satisfy certain properties. As in the single failure case, the computation of many of the replacement paths might be incomplete. The guarantee, however, would be that any $P(s,t,\{e_1,e_2\})$ replacement path whose last edge is not in $\FTMBFS(R \cup S)$ is fully computed by the algorithm, and it last edge is added to the output subgraph. 
The set of replacement paths which the algorithm attempts to compute correspond to the triplets:
$$Q_t=\{ (s,e_1,e_2) ~\mid~ e_1 \in P_{\sigma}(s,t,e_2) \mbox{~~and~~} e_2 \in P_{\sigma}(s,t,e_1), ~s \in S\}.$$
By the assumption (I1,I2), each vertex $t$ knows the last $2\sigma$ edges of the path $P(s,t,e)$ for every $s \in S$ and every $e \in \pi_{2\sigma}(s,t)$, it can compute $Q_t$ (see Claim \ref{cl:know}). Observe that $|Q_t|=|S| \cdot \sigma^2$. 
The algorithm starts by letting each vertex exchange its $Q_t$ set with its neighbors. The BFS tokens $\BFS(s,G \setminus \{e_1,e_2\})$ are permitted to pass from a vertex $u$ to a vertex $v$ only if $(s,e_1,e_2)\in Q_v$. 
To control the congestion due to the simultaneous constructions of multiple BFS trees, the vertices share a random string $\mathcal{SR}$. Each BFS algorithm $\BFS(s,G \setminus \{e_1,e_2\})$ for every $e_1,e_2 \in G$ and $s \in S$ starts in phase $\tau_{s,e_1,e_2}$ chosen uniformly at random in the range $\{1,\ldots, \Theta(|S| \cdot \sigma^2)\}$. Using the seed $\mathcal{SR}$ and the IDs of $s,e_1,e_2$, each vertex can compute $\tau_{s,e_1,e_2}$. These $\tau_{s,e_1,e_2}$ values are $O(\log n)$-wise independent. 

Each BFS algorithm $\BFS(s,G \setminus \{e_1,e_2\})$ then starts in phase $\tau_{s,e_1,e_2}$, and proceeds in a speed of one hop per phase. Each phase consists of $\ell=\Theta(\log n)$ rounds. The rules for passing the BFS tokens of $\BFS(s,G \setminus \{e_1,e_2\})$ are as follows:
\begin{itemize}
\item Each vertex $v$ that is not sensitive\footnote{In the analysis, we show that in the case where there is $u \in N(v)$ for which $(s,e_1,e_2)\in Q_u$, $v$ can indeed detect that it is not sensitive.} to $e_1,e_2$ sends the token $\BFS(s,G \setminus \{e_1,e_2\})$ in round $\tau_{s,e_1,e_2}+\dist(s,v,G \setminus \{e_1,e_2\})$ to every neighbor $u \in N(v)$ satisfying that $(s,e_1,e_2)\in Q_u$. 

\item Every vertex $v$ that is sensitive to $(s,e_1,e_2)$ upon receiving the first BFS token $\BFS(s,G \setminus \{e_1,e_2\})$ in phase $i$ does as follows:
\begin{itemize}
\item Let $w$ be the minimum-ID vertex in $N(v)$ from which $v$ has received the BFS token in that phase. Then, $v$ adds the edge $(w,v)$ to the output structure $H$.
\item $v$ sends the token $\BFS(s,G \setminus \{e_1,e_2\})$ in phase $i+1$ to every neighbor $u \in N(v)$ satisfying that $(s,e_1,e_2)\in Q_u$. 
\end{itemize}
\end{itemize}
This completes the description of the algorithm. 
\\
\paragraph{Analysis.} 
Let $Q'_t \subset Q_t$ be defined by $Q'_t=\{ (s,e_1,e_2) ~\mid~ e_1 \in P_{\sigma/2}(s,t,e_2) \mbox{~~and~~} e_2 \in P_{\sigma/2}(s,t,e_1), ~s \in S\}.$
Let $w$ be the first sensitive vertex (see Def. \ref{def:sen-det-two-f}) w.r.t. $(s,e_1,e_2)$ on the replacement path $P(s,t,\{e_1,e_2\})$. Recall that $SD(s,t,\{e_1,e_2\})=P(s,t,\{e_1,e_2\})[w,t]$ is the sensitive-detour of $P(s,t,\{e_1,e_2\})$.
\begin{observation}
Any vertex $w' \in SD(s,t,\{e_1,e_2\})$ is sensitive to the two edges $e_1,e_2$.
\end{observation}
\begin{proof}
Recall that $w$ is the first sensitive vertex on $P(s,t,\{e_1,e_2\})$, and thus the first vertex of the sensitive-detour.  Assume towards contradiction, that there exists a vertex $w' \in SD(s,t,\{e_1,e_2\})$ that is not sensitive to $(s,e_1,e_2)$. Let $P \in \{P(s,w',e_1), P(s,w',e_2)\}$ be such that $P=P(s,w',\{e_1,e_2\})$. 
By the uniqueness of the shortest paths, we have that $P(s,t,\{e_1,e_2\})[s,w']=P\circ P(s,t,\{e_1,e_2\})[w',t]$. We then have that $w \in P$ and thus $P[s,w]=P(s,w, \{e_1,e_2\})$ contradiction that $w$ is the first sensitive vertex on $P(s,t,\{e_1,e_2\})$. 
\end{proof}

\begin{claim}\label{cl:know}
By knowing (I1) and (I2), each vertex $t$ can compute the set $Q_t$.
\end{claim}
\begin{proof}
First assume that $e_1,e_2 \notin \pi_{\sigma}(s,t)$. We show that in this case $t$ can deduce that $(s,e_1,e_2) \notin Q_t$. If $e_1 \in P_{\sigma}(s,t,e_2)$ it must imply that $e_1 \in \pi(s,t)$ iff $e_1 \in \pi_{\sigma}(s,t)$.  In the same manner, if $e_2\in P_{\sigma}(s,t,e_1)$ it must imply that $e_2 \in \pi(s,t)$ iff $e_2 \in \pi_{\sigma}(s,t)$. Since $e_1,e_2 \notin \pi_{\sigma}(s,t)$, $t$ can conclude that $(s,e_1,e_2) \notin Q_t$.
Next assume that $e_1 \in \pi_{\sigma}(s,t)$. There are two subcases. If $e_2 \in P_{\sigma}(s,t,e_1)$, $q$ should check if also $e_1 \in P_{\sigma}(s,t,e_2)$. Since $e_2 \in P_{\sigma}(s,t,e_1)$ it must hold that $e_2 \in \pi(s,t)$ iff $e_2 \in \pi_{\sigma}(s,t)$. Thus, $t$ can verify if $e_2 \in \pi(s,t)$. If so, it has the path $P_{\sigma}(s,t,e_2)$. Otherwise, $P(s,t,e_2)$ is simply $\pi(s,t)$. 
The case where $e_2 \in  \pi_{\sigma}(s,t)$ is analogous. 
\end{proof}

To prove the correctness of the output structure, by Fact \ref{obs:dual-S} we need to show that $\LastE(P(s,t, \{e_1,e_2\}))$ is in $H$ for every $s \in S$ and every $e_1,e_2 \in E$. Throughout, we consider a fixed replacement path $P(s,t, \{e_1,e_2\})$ and assume w.l.o.g. that $e_1 \in \pi(s,t)$ and $e_2 \in P(s,t,e_1)$. 
\begin{lemma}\label{lem:caseone}
For every $(s,e_1,e_2) \notin Q'_t$ it holds that $\LastE(P(s,t, \{e_1,e_2\})) \in \FTMBFS(R)$.
\end{lemma}

\begin{proof}
Since $(s,e_1,e_2) \notin Q'_t$, it holds that either $e_2 \notin P_{\sigma/2}(s,t,e_1)$ or $e_1 \notin P_{\sigma/2}(s,t,e_2)$. First assume that $e_2 \notin P_{\sigma/2}(s,t,e_1)$. This implies that $|P(s,t,e_1)|\geq \sigma/2$ thus the  $(\sigma/4)$-length suffix of $P(s,t,e_1)$ contains a vertex $r \in R$ w.h.p. Since $|P(r,t,e_1)|\leq \sigma/4$, we have that $e_2 \notin P(r,t,e_1)$. Thus, $P(r,t,\{e_1,e_2\})=P(r,t,e_1)$, and thus $\LastE(P(s,t,\{e_1,e_2\}) \in \FTMBFS(R)$.
In the same manner, assume that $e_1 \notin P_{\sigma/2}(s,t,e_2)$. In this case, if $e_1 \notin P(s,t,e_2)$ then $P(s,t,e_2)=P(s,t,\{e_1,e_2\})$ and thus $\LastE(P(s,t,e_2))\in \FTMBFS(S)$. 
Next assume also that  $e_1 \in P(s,t,e_2) \setminus P_{\sigma/2}(s,t,e_2)$. 
This implies that $|P(s,t,e_2)|\geq \sigma/2$ thus the $(\sigma/4)$-length suffix of $P(s,t,e_2)$ contains a vertex $r \in R$. Since $|P(r,t,e_2)|\leq \sigma/4$, we have that $e_1 \notin P(r,t,e_2)$. Thus, $P(r,t,\{e_1,e_2\})=P(r,t,e_2)$, and thus $\LastE(P(s,t,\{e_1,e_2\}) \in \FTMBFS(R)$.
\end{proof}

To complete the correctness argument, it remains to show that $\LastE(P(s,t, \{e_1,e_2\})) \in H$ for every $(s,e_1,e_2) \in Q'_t$. We do it in two steps, depending on the length of the sensitive-detour. 
\begin{claim}
Let $(s,e_1,e_2) \in Q'_t$. If $|SD(s,t,\{e_1,e_2\})|\geq \sigma/3$, then $\LastE(P(s,t,\{e_1,e_2\}))\in H$. 
\end{claim}

\begin{proof}
Let $r \in R$ be a sampled vertex on $SD(s,t,\{e_1,e_2\})$. Since $r$ is sensitive, $P(s,r,\{e_1,e_2\})\neq P(s,r,e_1)$ and thus $e_2 \in P(s,r,e_1)$. Thus, letting $e_2=(x,y)$, both $r$ and $t$ belong to the subtree rooted at $y$ in $\BFS(s, G \setminus \{e_1\})$. Concluding that $P(r,t, \{e_1,e_2\})=P(r,t, e_1)$. Since the algorithm includes in $H$ the subgraph $\FTMBFS(R)$, we have that $\LastE(P(r,t, e_1))\in H$, the claim holds. 
\end{proof}

For now on, we consider replacement paths $P(s,t, \{e_1,e_2\})$ such that $(s,e_1,e_2) \in Q'_t$ and with a short sensitive-detour, i.e., $|SD(s,t,\{e_1,e_2\})|\leq \sigma/3$. We show the following. 
\begin{claim}
Let $(s,e_1,e_2) \in Q'_t$ and $|SD(s,t,\{e_1,e_2\})|\leq \sigma/3$. Then,
$(s,e_1,e_2) \in Q_{w'}$ for every $w' \in SD(s,t,\{e_1,e_2\})$.
\end{claim}
\begin{proof}
Fix $w'  \in SD(s,t,\{e_1,e_2\})$.
Since the detour is short it holds that $\dist(w',t, G \setminus \{e_1,e_2\})\leq \sigma/3$ for every $w' \in SD(s,t,\{e_1,e_2\})$. In addition, since $e_2 \in P_{\sigma/2}(s,t,e_1)$, we have that
\begin{equation}\label{eq:oneside}
\dist(e_2,w',G \setminus \{e_1\})\leq \dist(e_2,t,G \setminus \{e_1\})+\dist(t,w',G \setminus \{e_1,e_2\})\leq \sigma~.
\end{equation}
As $w'$ is sensitive, it holds that $e_2 \in P(s,w',e_1)$, and combining with Eq. (\ref{eq:oneside}) we have that 
$e_2 \in P_{\sigma}(s,w',e_1)$. In the same manner, since $e_1 \in P_{\sigma/2}(s,t,e_2)$ and $e_1 \in P(s,w',e_2)$, by the same reasoning we have that $e_1 \in P_{\sigma}(s,w',e_2)$.
We conclude that $(s,e_1,e_2) \in Q_{w'}$. 
\end{proof}
We next show that the BFS token $\BFS(s, G \setminus \{e_1,e_2\})$ arrives each vertex $w' \in SD(s,t,\{e_1,e_2\})$ in phase $\dist(s,w', G \setminus \{e_1,e_2\})+\tau_{s,e_1,e_2}$. Since for every vertex $w' \in SD(s,t,\{e_1,e_2\})$ it holds that $(s,e_1,e_2)\in Q_{w'}$, it is guaranteed that the BFS token $\BFS(s, G \setminus \{e_1,e_2\})$ arriving $w'$ in $SD(s,t,\{e_1,e_2\})$ in phase $i$ is sent to the next hop $w'' \in SD(s,t,\{e_1,e_2\})$ in phase $i+1$. Therefore it is sufficient to show that the first vertex, say $w$, on the sensitive-detour $SD(s,t,\{e_1,e_2\})$ receives the token $\BFS(s, G \setminus \{e_1,e_2\})$ in phase $\dist(s,w', G \setminus \{e_1,e_2\})+\tau_{s,\{e_1,e_2\}}$. Let $q$ be the neighbor of $w$ on $P(s,t,\{e_1,e_2\})$ not in $SD(s,t,\{e_1,e_2\})$. 
\begin{claim}\label{cl:cancomp}
$q$ sends to $w$ (first vertex on the sensitive-detour) the BFS token $\BFS(s, G \setminus \{e_1,e_2\})$ in phase $\dist(s,w, G \setminus \{e_1,e_2\})+\tau_{s,e_1,e_2}$. 
\end{claim}
\begin{proof}
By the description of the algorithm, it is sufficient to show that $q$ knows that (i) it is not sensitive to $e_1,e_2$ and (ii) its distance $\dist(s,q,G \setminus \{e_1,e_2\})$. 

We first claim that if $q$ would have been sensitive to $(s,e_1,e_2)$ then it must have hold that $e_1 \in P_{2\sigma}(s,t,e_2)$ and $e_2 \in P_{2\sigma}(s,t,e_1)$. To show this, assume that $q$ is sensitive to $(s,e_1,e_2)$ and thus $e_1 \in P(s,q,e_2)$ and $e_2 \in P(s,q,e_1)$.
Since $e_1 \in P_{\sigma}(s,w,e_2)$ and $e_2 \in P_{\sigma}(s,w,e_1)$, by the triangle inequality we have that
$$\dist(e_1, q,G\setminus \{e_2\})\leq \dist(e_1, w,G\setminus \{e_2\})+1\leq \sigma+1~.$$
We therefore conclude that by assumption (I1,I2), $q$ knows that it is not sensitive to $(s,e_1,e_2)$. 
Next, we show that $q$ knows the distance $\dist(s,q,G\setminus \{e_1, e_2\})$.
Assume first that $e_1 \notin P(s,q,e_2)$. If $e_2 \in \pi(s,q)$, then it must be that $e_2 \in \pi_{2\sigma}(s,q)$. Thus $q$ can tell if  $e_2 \in \pi(s,q)$ and by assumption (I1), $q$ knows $\dist(s,q,G\setminus \{e_2\})$ which in this case equals to  $\dist(s,q,G\setminus \{e_1, e_2\})$. Otherwise, if $e_2 \notin \pi(s,q)$, we have that $\dist(s,q,G)=\dist(s,q,G\setminus \{e_1, e_2\})$. The proof works analogously when assuming that $e_2 \notin P(s,q,e_1)$. We conclude that by knowing (I1,I2), $q$ can compute the distance $\dist(s,q,G\setminus \{e_1, e_2\})$. 
\end{proof}
\begin{corollary}
For every path $P(s,t,\{e_1,e_2\})$ satisfying that (i) $(s,e_1,e_2)\in Q'_t$ and (ii) $|SD(s,t, \{e_1,e_2\})|\leq\sigma/3$, it holds that the detour $SD(s,t, \{e_1,e_2\})$ is fully computed by the algorithm (i.e., the BFS token propagates through all the vertices on the sensitive-detour). Consequently, $\LastE(P(s,t, \{e_1,e_2\}))\in H$. 
\end{corollary}

\paragraph{Round Complexity.} We next analyze the round complexity. 
The computation of the structure $\FTMBFS(R \cup S)$ takes $O(\sqrt{(|R|+|S|)n}+D)=\widetilde{O}(n^{7/8}\cdot |S|^{1/8}+D)$ rounds.
Running the truncated BFS trees $\BFS(s,G\setminus \{e_1,e_2\})$ takes in total $\widetilde{O}(D+(\sigma_2)^2\cdot |S|)$. Combining with the round complexity of Lemma \ref{lem:short-dist} yields the desired bound of $\widetilde{O}(D+|S|^{5/4}n^{3/4}+|S|^{1/8}n^{7/8})$. 
\\
\paragraph{Size.} The total number of edges in $\FTMBFS(R \cup S)$ is bounded by $O(\sqrt{|R|+|S|}\cdot n^{3/2})$.
In addition, each vertex $t$ adds at most $|Q_t|=O(|S|\cdot \sigma^2)$ edges to $H$. 
Plugging $\sigma=(n/|S|)^{1/4}$ and $|R|=O(n\log n/\sigma)$ yields the desired edge bound of $\widetilde{O}(|S|^{1/8}n^{15/8})$.  

\subsection{Learning Distances and Short RP Segments of Near Faults}\label{sec:short-dist}
In this subsection we fill in the missing piece of the algorithm by proving Lemma \ref{lem:short-dist}, and thus establishing Theorem \ref{thm:FT-MBFS-dist-dual}. The computation of the information (I1,I2) for every vertex $t$ is done in two key steps depending on the structure of the $P(s,t,e)$ path.

A replacement-path $P(s,t,e)$ for $e \in \pi_{\sigma_2}(s,t)$ is said to be \emph{easy} if $|SD(s,t,e)|\leq \sigma_1$. Otherwise, the path $P(s,t,e)$ for $e \in \pi_{\sigma_2}(s,t)$ is \emph{hard}. 
\\
\paragraph{Computing the information for easy replacement paths.}
We will present a somewhat stronger algorithm that computes (I1,I2) for every $P(s,t,e)$ paths satisfying that
$e \in \pi_{\sigma_1}(s,t)$ (rather than just $e \in \pi_{\sigma_2}(s,t)$).
The algorithm simply applied the second step of the single-failure \FTMBFS\ algorithm with parameter $\sigma=8\sigma_1$. Recall that in this phase, a partial collection of replacement paths is computed which is characterized by the given parameter $\sigma$. By the proof of Lemma \ref{lem:main-cor} (Case (3)), we have that each $t$ knows $\dist(s,t,G \setminus \{e\})$ for every \emph{easy} replacement path. It therefore remains for it to learn also the $\sigma_2$-length suffix of these paths. We next show that this can be done by a simple extension of the algorithm.

\begin{claim}\label{cl:easy-easy}
Within extra $\widetilde{O}(D+\sigma_1 \cdot \sigma_2 \cdot |S|)$ rounds, 
every vertex $t$ can learn the $\sigma_2$-length suffix of every easy replacement path $P(s,t,e)$, for every $e \in \pi_{\sigma_1}(s,t)$ and $s \in S$.
\end{claim}
\begin{proof}
By applying the standard \FTMBFS\ algorithm with parameter $8\sigma_1$, every vertex $t$ receives the BFS token $\BFS(s, G \setminus \{e_1\})$ provided that $e \in \pi_{\sigma_1}(s,t)$ and $|SD(P(s,t,e))|\leq \sigma_1$. For every vertex $t$, let 
$$\widetilde{Q}_t=\{(s,e) ~\mid~ e \in \pi_{\sigma}(s,t) \mbox{~and~} t \mbox{~received~the token~} \BFS(s, G \setminus \{e\})\}.$$ Each $t$ exchange this set with its neighbors. Note that by Case (3) of Lemma \ref{lem:main-cor}, the set of easy paths are contained in $\widetilde{Q}_t$ (but this set might contain even $(s,e)$ pairs that correspond to hard replacement-paths).

Next, we run a modified variant for each BFS algorithm $\BFS(s,G \setminus \{e\})$ in which each vertex also learns its last $2\sigma_2$ edges on its BFS path from $s$. A single modified BFS algorithm still runs in $O(D)$ rounds, but it passes $O(\sigma_2)$ messages per edge, rather than $O(1)$ many messages as in the standard BFS computation. 
The modified BFS algorithms $\BFS(s,G \setminus \{e\})$ for each $s$ and $e$ are computed simultaneously using the random delay technique, passing the $\BFS(s,G \setminus \{e\})$ messages according to the same rules as in the \FTMBFS\ algorithm. The only difference is that each BFS algorithm sends now $\sigma_2$ messages on every edge, rather than a constant number of messages (as in a standard BFS computation). 

Since each vertex $t$ needs to get information from $\sigma_1 \cdot |S|$ modified BFS algorithms, and from each such algorithm it needs to receive $\sigma_2$ messages (corresponding to its $\sigma_2$ last edges on its path from the root), the total edge congestion is bounded by $O(\sigma_1 \cdot \sigma_2 \cdot |S|)$. Using the random delay approach, this can be done in $\widetilde{O}(D+\sigma_1 \cdot \sigma_2 \cdot |S|)$ rounds, w.h.p. 
\end{proof}

\paragraph{Computing the information for hard replacement paths.}
It remains to consider the hard replacement paths $P(s,t,e)$. I.e., paths for which $e \in \pi_{\sigma_2}(s,t)$ and their sensitive-detour is of length at least $\sigma_1$. Unlike the previous algorithm, here we might not learn the distances $\dist(s,t,G\setminus \{e\})$ for edges $e \in \pi_{\sigma_1}(s,t)\setminus \pi_{\sigma_2}(s,t)$. We assume that this step is applied after already computing the information for the easy replacement paths.

Let $R$ be a random sample of $O(n\log n/\sigma_1)$ vertices. The algorithm computes BFS trees $T_r=\BFS(r,G)$ for every $r \in R$. In addition, each vertex also learns its last $\sigma_2$ edges on each $\pi(r,t,T_r)$ paths. Using the random delay approach, this can be done in $\widetilde{O}(D+(n\log n/\sigma_1)\cdot \sigma_2)$ rounds.
\begin{lemma}
One can compute LCA (Least Common Ancestor) labels in each BFS tree $T_s$, $s \in S$ in total time $\widetilde{O}(D+S)$. The size of each LCA label is $O(\log^2 n)$ bits (per tree $T_s$).
\end{lemma}
\begin{proof}
Computing the LCA labels for single tree $T_s$ can be done in $\widetilde{O}(D)$ rounds and sending $\widetilde{O}(1)$ along each tree edge.
This can be done by computing the heavy-light decomposition and sending each vertex a compressed representation of its path from the root. Ideas along this line appear in \cite{dory2019improved}. 
To compute LCA w.r.t. $S$ trees simultaneously, we simply apply the random delay approach. 
\end{proof}
Consider an hard replacement-path $P(s,t,e)$ and let $q$ be the neighbor before $w$ on the path, where $w$ is the first sensitive vertex on $P(s,t,e)$. Let $e=(x,y)$. We claim the following, see Fig. \ref{fig:hard-path} for an illustration. 
\begin{claim}\label{cl:save}
For every hard replacement path $P(s,t,e)$, there must be two vertices $r_1,r_2$ such that (i) $\dist(r_1,r_2,G)\leq \sigma_1/16$, (ii) $r_1$ is not sensitive to $e$ and $r_2$ is sensitive to $e$ and (iii) $e \notin \pi(r_1,r_2)$.
\end{claim}
\begin{proof}
We claim that for every vertex $w'$ appearing on the $(\sigma_1/8)$-length prefix of $SD(s,t,e)$ it holds that 
$e \notin \pi_{\sigma_1/8}(s,w')$. Assume towards contradiction otherwise, since $e \in \pi_{\sigma_1/8}(s,w')$ and $e \in \pi_{\sigma_2}(s,t)$ (and $\sigma_2<<\sigma_1$), the tree path between $w'$ and $t$ in $T_s$ is free from $e$ and has length at most $\sigma_1/4$. As $\dist(w',t, G \setminus \{e\})=|SD(s,t,e)[w',t']|\geq  \sigma_1/2$, we end with a contradiction. 

Next, let $w$ be the first vertex on $SD(s,t,e)$ and let $q$ be the vertex that appears just before $w$ on $P(s,t,e)$. By the uniqueness of the shortest-path, $P(s,t,e)=\pi(s,q)\circ P(s,t,e)[q,t]$. We now claim that $\dist(s,q,G)\geq \sigma_1/8-1$. Since $e \notin \pi_{\sigma_1/8}(s,w)$, it implies that $\dist(s,w,G)\geq \sigma_1/8$. Since $(q,w)$ is a non-tree edge in the BFS tree $T_s$, we conclude that $\dist(s,q,G)\geq \sigma_1/8-1$. 

Therefore the $\sigma_1/32$ suffix of $\pi(s,q)$ contains a vertex $r_1 \in R$ that is not sensitive to $e$. 
The $\sigma_1/32$ prefix of the sensitive-detour $SD(s,t,e)$ contains a vertex $r_2 \in R$ that is sensitive to $e$. 
Since the distance between $r_1,r_2$ on $P(s,t,e)$ is at most $\sigma_1/16$ and since $\dist(e,r_2,G)\geq \sigma_1/8$, we conclude that $\dist(r_1,r_2,G)=\dist(r_1,r_2,G \setminus \{e\})$.
\end{proof}
\begin{figure}[h!]
\begin{center}
\includegraphics[scale=0.30]{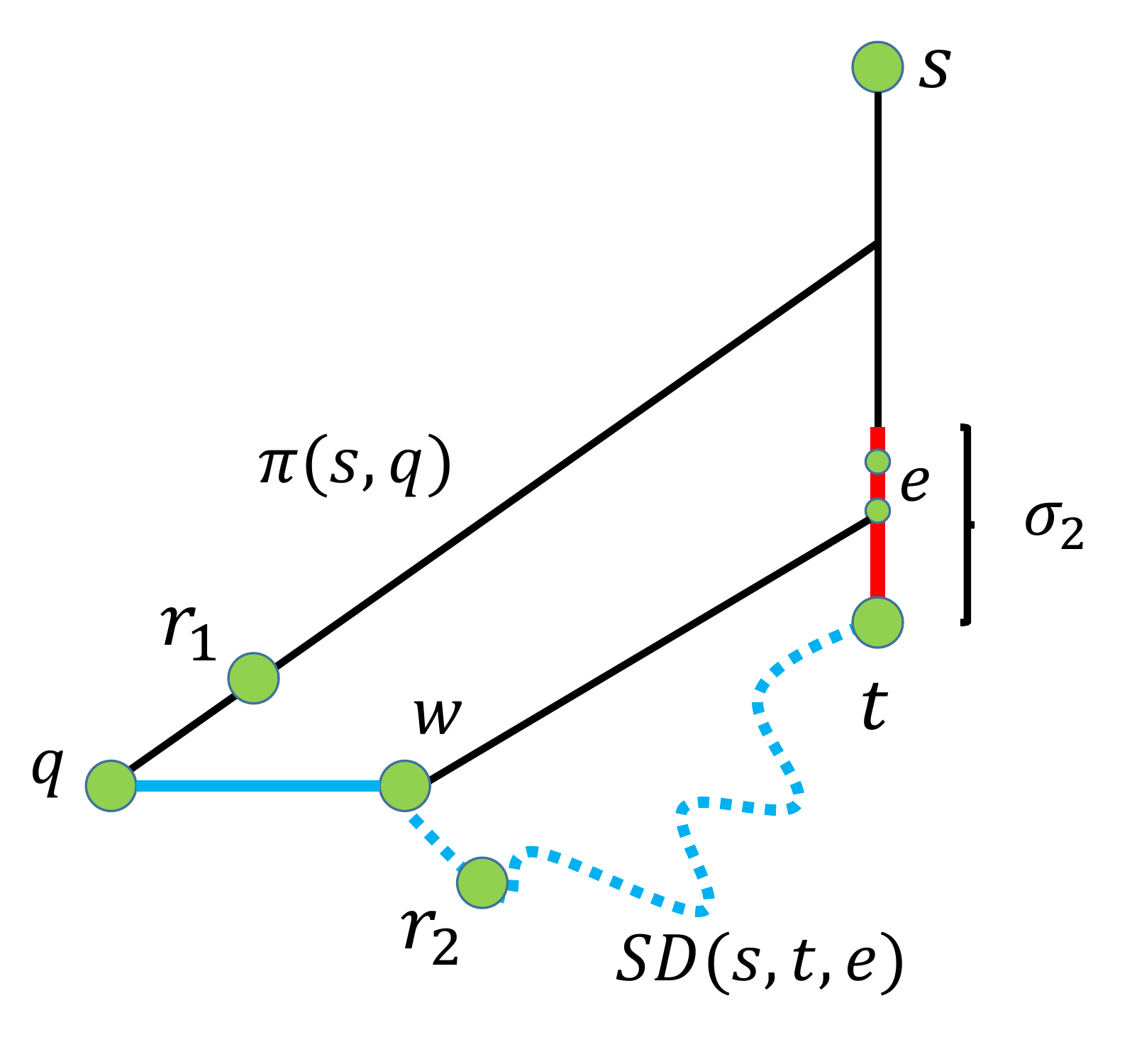}
\caption{ \label{fig:hard-path} An illustration for the proof of Claim \ref{cl:save}. Shown in an hard $P(s,t,e)$ path where $e=(x,y) \in \pi_{\sigma_2}(s,t)$. The vertex $w$ is the first vertex on the sensitive-detour, thus the entire $P(s,t,e)[w,t]$ is contained in the vertex set of $T_s(y)$, where is the subtree of $T_s$ rooted at $y$. Dashed edges correspond to the path segment $SD(s,t,e)$. Since $e$ is very close to $t$, but $e$ is somewhat far from the vertices on the prefix of the sensitive-detour, there are two vertices $r_1,r_2$ that satisfy the properties of the claim.}
\end{center}
\end{figure} 
The algorithm then lets each vertex $r$ in $R$ send to all vertices in the graph the following:
\begin{itemize}
\item The list of the distances $\dist(r,r',G)$ for every $r'$ in $R$.
\item The $\widetilde{O}(1)$-length bit LCA label of $r$ in each tree $T_s$.
\end{itemize}
Overall, the total information sent is $\widetilde{O}(|R|^2+|S|\cdot |R|)$. This can be done in 
$\widetilde{O}(|R|^2+|S|\cdot |R|+D)$ rounds by a simple pipeline.

Now every vertex $t$ is doing the following calculations for every edge $e \in \pi_{\sigma_1}(s,t)$ for which it did not receive a BFS token $\BFS(s,G \setminus \{e\})$ in the first phase of the algorithm (of handling the easy replacement paths). Using the LCAs of all vertices in $R$ with respect to $T_s$, it computes the set $R^+_{e}$ and $R^-_e$ where $R^-_{e}=\{r \in R ~\mid~ e \notin \pi(s,r)\}$ and $R^+_e=R \setminus R^-_{e}$. 
Note that $e \in \pi(s,r)$ only if the LCA of $r$ and $t$ is \emph{below} the failing edge $e$. Since $t$ has the $2\sigma_1$-length suffix of its $\pi(s,t)$ path, it can detect if the LCA is below the edge $e$. Let 
$$\dist(s,t,G \setminus \{e\})=\min_{r_1 \in R^-_{e}} ~~\min_{r_2 \in R^+_e, \dist(r_1,r_2,G)\leq \sigma_1/16} \dist(s,r_1,G)+\dist(r_1,r_2,G)+\dist(r_2,t,G)~.$$

Let $r^*_1 \in R^-_{e}$ and $r^*_2 \in R^+_{e}$ be the vertices that minimize the $\dist(s,t,G \setminus \{e\})$.
Then, the $t$ lets $P_{\sigma_2}(s,t,e)=\pi_{\sigma_2}(r^*_2,t)$. 
This completes the description of the algorithm, we now complete proof of Lemma \ref{lem:short-dist}.
\begin{proof}[Proof of Lemma \ref{lem:short-dist}]
The correctness of the easy replacement paths follows immediate. Since each vertex $t$ knows $\pi_{\sigma_1}(s,t)$ for every source $s$, it can detect the set of hard replacement paths in the following sense. 
For every pair $s,e$, for which $t$ has received the BFS token $\BFS(s,G\setminus \{e\})$, $t$ learns the segment $P_{\sigma_2}(s,t,e)$ of this path using the algorithm for the easy paths. Thus, at this point, $t$ can conclude that any pair $s,e$ for which $e \in \pi_{\sigma_1}(s,t)$ but $t$ did receive the $\BFS(s,G\setminus \{e\})$ token, it must hold that the sensitive-detour of this path is long, and thus the replacement is hard (see Case (2) of Lemma \ref{lem:main-cor}). The correctness of computing the hard RPs distances follows by Claim \ref{cl:save}.

The running time for computing the information for the easy RPs is $\widetilde{O}(D+|S|\cdot \sigma_1\cdot \sigma_2)$. It is easy to see that the total running time of for computing the information for hard RP is 
$\widetilde{O}(D+|R|^2+|S|\cdot |R|)=\widetilde{O}(D+(n/\sigma_1)^2+|S|\cdot (n/\sigma_1))$. 
\end{proof}

\subsection{Fault Tolerant Additive Spanners}

\paragraph{Single Fault.} We next show that using the \FTMBFS\ construction yields a $+2$ FT-additive spanner with 
$\widetilde{O}(n^{5/3})$ edges, which fits the state-of-the-art size of these structures (up to poly-logarithmic factors). 
We next prove Cor. \ref{cor:two-add-onef}.
\begin{proof}[Proof of Cor. \ref{cor:two-add-onef}]
Let $S=\Sample(V, 10\log n\cdot n^{-2/3})$ be a random sample of vertices. 
The subgraph $H$ consists of the following edges: edges incident to vertices of degree at most $n^{2/3}$, as well as an \FTMBFS\ subgraph $H'$ w.r.t. the sources $S$. 
This subgraph can be computed in $\widetilde{O}(D +\sqrt{n |S|})=\widetilde{O}(D+n^{5/6})$ rounds. 

To see that $H$ is an $+2$ FT-additive spanner, consider a replacement path $P(s,t,e)$ for some $s,t \in V$ and $e \in E$. Let $u$ be the first vertex (closest to $s$) on the path with a missing edge. This implies that the degree of $u$ in $G$ is at least $n^{2/3}$, and therefore w.h.p. $u$ has at least two neighbors in $S$. Thus, there exists $s' \in S \cap N(u)$ such that $(u,s) \in H \setminus \{e\}$. Since the algorithm adds the \FTBFS\ structure $H_{s'}$ w.r.t. $s'$ we have that 
\begin{eqnarray*}
\dist(s,t, H \setminus \{e\})&=&\dist(s,u, H \setminus \{e\})+\dist(u, t,H \setminus \{e\})
\\&\leq& \dist(s,u, G \setminus \{e\}) + 1+ \dist(s', t,H \setminus \{e\}) \\
& \leq & \dist(s,u, G \setminus \{e\}) + 1+ \dist(s', t,G \setminus \{e\})\\
&\leq & \dist(s,u, G \setminus \{e\}) +\dist(u, t,G \setminus \{e\})+2=\dist(s,t, G \setminus \{e\})+2~.
\end{eqnarray*}
\end{proof}

\paragraph{Two Faults.} In the same manner, the dual-failure \FTMBFS\ structures can be used to provide $+2$ dual-failure FT-additive spanners. We next prove Cor. \ref{cor:two-add-twof}. 

\begin{proof}[Proof of Cor. \ref{cor:two-add-twof}]
Let $S$ be a random subset of $O(n^{1/9}\log n)$ vertices.
A vertex $t$ is high-deg if  $\deg(t,G)\geq 10n^{8/9}$ and otherwise, it is low-deg.
It is easy to see that w.h.p. every high-deg vertex $t$ with $\deg(t,G)\geq 10n^{8/9}$ has at least \emph{three} sampled neighbors in $S$. We then connect each such vertex $t$ to three representatives neighbors in $S$. 
Next, the algorithm adds the edges incident to all the low-deg vertices. Finally, the algorithm adds the dual-failure \FTMBFS\ w.r.t. $S$. 

We next show that the final subgraph $H$ is an $+2$ FT-additive spanner. Consider a replacement path $P(u,v,\{e_1,e_2\})$ and let $t$ be the first high-deg vertex on the path, closest to $u$. W.h.p., $t$ is connected to some neighbor $s \in N(u) \cap S$ in the spanner. Since $e_1,e_2$ fail and $t$ has w.h.p. three edges to neighbors in $S$, at least one of these edges (say to $s$) survive in $H \setminus \{e_1,e_2\}$. Since $H$ contains a dual-failure \FTBFS\ structure w.r.t. $s$, we have that:

\begin{eqnarray*}
\dist(s,t, H \setminus \{e_1,e_2\})&=&\dist(u,t, G \setminus \{e_1,e_2\})+1+\dist(s, v,H \setminus \{e_1,e_2\})
\\&\leq& \dist(u,t, G \setminus \{e_1,e_2\})+1+\dist(s, v,G \setminus \{e_1,e_2\}) \\
& \leq & \dist(s,u, G \setminus \{e_1,e_2\}) +\dist(u, v,G \setminus \{e_1,e_2\})+2\\
&\leq & \dist(u,v, G \setminus \{e_1,e_2\})+2~.
\end{eqnarray*}
The round complexity is dominated by the computation of the dual-failure \FTMBFS\ structure which takes 
$\widetilde{O}(D+n^{3/4}\cdot|S|^{5/4})=\widetilde{O}(D+n^{8/9})$ rounds. The total number of edges is 
$\widetilde{O}(n^{15/8}\cdot S^{1/8})=\widetilde{O}(n^{17/9})$.
\end{proof}

\newpage
\bibliographystyle{plainurl}
\bibliography{ref}

\end{document}